\documentclass[11pt]{article}

\usepackage{amsmath}
\usepackage{amssymb}
\usepackage{latexsym}
\usepackage{epsfig}
\usepackage{fullpage}
\usepackage{url}
\usepackage[usenames]{color}
\usepackage{microtype}

\usepackage{cite}

\newtheorem{lemma}[figure]{Lemma} 
\newtheorem{theorem}[figure]{Theorem} 
\newtheorem{definition}[figure]{Definition}

\newcommand{\etal}{et al.}

\newenvironment{proof}{\noindent\textbf{Proof: }\ignorespaces}
  {\hspace*{\fill}$\Box$\medskip}

\newcommand\marrow{{\marginpar[\hfill$\longrightarrow$]{$\longleftarrow$}}}
\newcommand{\remark}[3]{\textcolor{blue}{\textsc{#1 #2:}}
\textcolor{red}{\marrow\textsf{#3}}}
\renewcommand{\remark}[3]{}

\newcommand{\eqdef}{:=}

\title{A Static Optimality Transformation with Applications\\ 
to Planar Point Location}

\author{
John Iacono \and Wolfgang Mulzer
}

\begin{document}

\clubpenalty=10000
\widowpenalty = 10000

\maketitle

\begin{abstract} 
Over the last decade, there have been several data
structures that, given a planar subdivision and a probability 
distribution over the 
plane, provide a way for answering point location queries that is 
fine-tuned for the 
distribution. All these  methods suffer from the requirement that the query
distribution must be known in advance. 

We present a new data structure for point location queries in 
planar triangulations. Our structure is asymptotically as
fast as the optimal structures, but it requires no prior
information about the queries. 
This is a \textsc{2-d} analogue of the jump from Knuth's
optimum binary search trees (discovered in 1971) to the splay trees of 
Sleator and Tarjan in 1985. While the former need to know the
query distribution, the latter are \emph{statically optimal}.
This means that we 
can adapt to the query sequence and 
achieve the same asymptotic performance as an
optimum static structure, without needing any 
additional information.  
\end{abstract}

%\category{F.2.2}{Algorithms}{Nonnumerical Algorithms and Problems}

%\terms{Algorithms}

\section{Introduction}

We consider the problem of finding a statically optimal data structure
for planar point location in triangulations. This problem and related
problems have a long history that goes back to the dawn of computer science.
Thus, before giving a formal description of the problem and of our results,
let us first provide some background on the history and motivation behind
our work.

\subsection{{1-D} History}

Comparison-based predecessor search constitutes one of the oldest problems
in computer science: given a set $S$ from a totally ordered universe
$U$, we would like to construct a data structure for answering 
\emph{predecessor queries}. 
In such a query, we are given an element $x \in U$, and we need to return 
the largest $y \in S$ with $y \leq x$ (or $-\infty$, if no
such $y$ exists). In the most general \emph{decision-tree} model, 
we are allowed to evaluate 
in each step an \emph{arbitrary} function $f: U \rightarrow \{0,1\}$ on $x$,
where the choice of $f$ may depend on the outcomes of the previous evaluations.
The classic solution sorts $S$ during preprocessing and
answers queries in $O(\log n)$ steps through binary search, where
$n$ denotes the size of $S$.
Information theoretic arguments imply that any such comparison-based algorithm 
requires $\Omega(\log n)$
steps in the worst case (see, e.g., Ailon \etal~\cite[Section~2]{AilonChClLiMuSe11}
for more details).

However, the story does not end here. Early in the history of computer 
science, researchers realized that if the distribution of query outcomes
is sufficiently biased,  $o(\log n)$ expected-time query processing 
becomes possible. 
This insight led to the invention of \emph{optimal search trees}.
These are specialized data structures for the case
that the query outcomes are drawn independently from a known
fixed distribution, and a wide literature studying their variants 
and extensions have been
developed~\cite{optimum1,optimum2,optimum3,optimum4,optimum5,optimum6, optimum7,
optimum8, optimum9, optimum10, optimum11, optimum12, optimum13, optimum14,
optimum15}.  In this context, optimality is characterized by the \emph{entropy} 
of the distribution: if $p_i$ denotes the probability of the $i$th outcome, 
the entropy $\mathcal{H}$ is defined as 
$\sum_i -p_i \log_2 p_i$. Information theory~\cite{weaver} 
shows that $\mathcal{H}$  
is a lower bound for the expected number of steps that any 
comparison-based algorithm needs to answer a predecessor query,
assuming that the searches are drawn independently from a 
fixed distribution (e.g.,~\cite[Claim~2.2]{AilonChClLiMuSe11}).

All the above results require that the distribution, or a suitable
approximation thereof, be known in advance. This situation changed in 
1985, when Sleator and Tarjan~\cite{splay} introduced \emph{splay trees}. 
These trees have
many amazing properties, not the least of which is 
called \emph{static optimality}. This means that for any
sufficiently long query sequence, splay trees are
asymptotically as fast as optimal static search trees. For this, splay trees 
require no prior information on the query distribution.

\subsection{2-D History}

Planar point location is a fundamental problem in computational
geometry. A \emph{triangulation} $S$ is a partition of the plane into
(possibly infinite) triangles. Given $S$, 
we need to construct a data structure for \emph{point location queries}: 
given a point $p \in \mathbb{R}^2$, return the triangle of $S$ that contains it.
Again, we use a decision-tree model. This means that in each step we may
evaluate an arbitrary function 
$f: \mathbb{R}^2 \rightarrow \{0,1\}$  on $p$, where $f$ may depend on 
the previous comparisons.

There are several point location structures with $O(\log n)$ query time, which
is optimal in our decision-tree model.
These structures are notable not only for achieving optimality,
but for doing so through very different methods,
such as planar separators~\cite{LiptonTa80,LiptonTa79}, Kirkpatrick's 
successive refinement approach~\cite{kirk}, persistence~\cite{ppl5},
layered DAGs~\cite{EdelsbrunnerGuSt86},
or randomized incremental construction~\cite{trapseidel,Mulmuley90}.

Once again, it makes sense to consider biased
query distributions.
For a known fixed distribution of point location queries, there are
several data structures that achieve optimal expected query time,
assuming independence. 
These \emph{biased} structures are
analogous to optimal search trees. Thus, we can use the same information
theoretic arguments to characterize the optimal expected query time
by the entropy $\mathcal{H}$ of the probabilities of the 
queried regions~\cite[Claim~2.2]{AilonChClLiMuSe11}.
 
A series of papers by 
Arya~\etal~\cite{AryaMaMo07,AryaMaMoWo07,
AryaMaMo01,AryaMaMo01a,AryaMaMo00,AryaChMoHa00}
converge on two algorithms. The first one achieves query time 
$\mathcal{H}+O(\sqrt{\mathcal{H}}+1)$ with
$O(n)$ space, while the second, simpler, algorithm supports queries in time $(5 \ln
2)\mathcal{H}+O(1)$ and $O(n \log n)$ space.\footnote{In this context,
\emph{query time} refers to the expected depth of the associated
decision tree.} The latter algorithm is a 
truly simple variant of randomized incremental 
construction~\cite{trapseidel,Mulmuley90}, where the 
random choices are biased according to the distribution. Both structures 
are randomized and have superlinear construction costs. 
Iacono~\cite{Iacono04} presented a data structure 
that supports
$O(\mathcal{H})$ time queries in $O(n)$ space, but, 
unlike the aforementioned results, it is
deterministic, can be constructed in linear time, and has terrible constants.

\subsection{Creating a point location structure that is statically optimal}

In view of the developments for binary search trees, one 
question presents itself:  Is there a point location structure that is
asymptotically as fast as the biased structures, \emph{without} explicit
knowledge of the query distribution? Or, put differently, can a point
location structure achieve a running time similar to the static optimality bound
of splay trees? This open problem, which we resolve here, explicitly appears in
several previous works on point location, e.g., in
Arya~et al.~\cite[Section 6]{AryaMaMoWo07}:

\begin{quote}
Taking this in a different direction, suppose that the query
distribution is not known at all. That is, the probabilities that the query
point lies within the various cells of the subdivision are unknown. In the
1-dimensional case it is known that there exist self-adjusting data structures,
such as splay trees, that achieve good expected query time in the limit.
Do such self-adjusting structures 
exist for planar point location?
\end{quote}

There are several possible approaches towards statically optimal
point location. One, suggested above, would be to create 
some sort of
self-adjusting point location structure and to analyze it in a way similar to
splay trees. This has not been done; we suspect that the main stumbling block is
that all known efficient structures are 
\emph{comparison DAGs}~\cite{kirk,ppl5,EdelsbrunnerGuSt86,Mulmuley90,trapseidel}: 
they can 
be represented as a directed acyclic graph with a unique source and 
out-degree $2$, such that each node corresponds to a planar region. A 
point location
query proceeds by starting at the source and by following in each step
an edge that is determined by comparing the query point with a fixed line.
The query continues until it reaches a sink, whose corresponding region
constitutes the desired query outcome. In order to achieve 
reasonable space
usage, it seems essential to use a DAG instead of  a simple tree.
Unfortunately, we do not know how to  perform 
rotation-like local changes in such DAGs that would mimic the behavior of
splay trees.

Another possible avenue is to use splay trees in an
existing structure. Goodrich~\etal~\cite{goodrich2} followed this approach,
using essentially a hybrid of splay trees and the persistent line-sweep
method. Unfortunately, their method does not give a result
optimal with respect to the entropy of the original distribution of
query outcomes,
but rather to the entropy of the probabilities of querying
regions of a \emph{strip decomposition} of the triangulation. 
The latter is obtained
by drawing vertical lines through every point of the triangulation. 
This strip decomposition could split a high-probability triangle 
into several parts
and could potentially increase the entropy of the query result by 
$\Omega(\log n)$, the worst possible; see Figure~\ref{fig:goodrich} 
for an example.
\begin{figure}[ht]
\begin{center}
\includegraphics{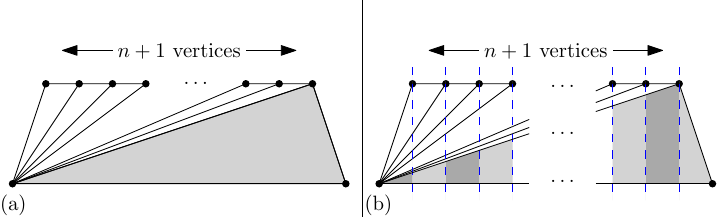}
\end{center}
\caption{A bad example for the strip decomposition of~\cite{goodrich2}: 
(a) We have $n+3$ vertices
and $n+1$ triangles. The small triangles each have query probability
$1/n^2$, the large, shaded, triangle has query probability $1-1/n$.
The entropy is $(1/n)\log n^2 + (1-1/n)\log(n/(n-1)) = O(1)$. 
(b) The strip decomposition partitions the large
triangle into $n+1$ parts. Suppose each part has probability 
$(n-1)/n(n+1) \approx 1/n$. The resulting entropy is
larger than $(1-1/n)\log(n(n+1)/(n-1)) = \Omega(\log n)$.}

\label{fig:goodrich}
\end{figure}

One might also try to create a structure with the 
\emph{working set property}. This property, 
originally used in the analysis of splay
trees, states that the processing time for a query $q$ is logarithmic
in the number
of distinct queries since the last query that returned the same result
as $q$. 
The working set property implies static 
optimality~\cite{Iacono01a}; 
it has also proved useful in several other
contexts \cite{BoseDoDuHo10,Elmasry06,BoseDoLa08,Iacono05,Iacono00}. 
Most importantly, there is a general
transformation from a \emph{dynamic} $O(\log n)$ time 
structure into one with the working
set property~\cite{Iacono01a}. 
Unfortunately, even though several dynamic data structures for 
predecessor searching are known (e.g., AVL trees \cite{avl} or
red-black trees \cite{redblack}), it remains a prominent open problem
to develop a point location structure that supports insertions, deletions,
and queries in $O(\log n)$ time.
(Note that \cite{Talib:1996:TME:792755.792810} claims to modify Kirkpatrick's
method to allow for $O(\log n)$ time insertions, deletions, and queries. The
claimed result is wrong.~\footnote{The method presented makes the assumption 
that given a triangle $T$ in a triangulation of size $n$ on which
a Kirpatrick hierarchy has been built, the complexity of the 
intersection of $T$ with any level of the hierarchy is constant; this is 
false as examples where the intersection is size $\sqrt{n}$ are easy to produce.})

Our solution to the problem of statically optimal point location is
very simple: we take a biased structure that needs to be initialized with
distributional information, and we rebuild it periodically using the observed
frequencies for each region. We do not store all the regions in the biased
structure---this would make the rebuilding step too expensive.
Instead, upon rebuilding we create a structure storing only the $n^\beta$ most
frequent items observed so far,
where $\beta \in (0,1)$ is some suitable constant.
We resort to a static
$O(\log n)$ time structure to
complete queries for the remaining regions.
 The rebuilding takes place after every $n^\alpha$ queries, for some constant
$\alpha \in (\beta,1-\beta)$.
This is a
simple and general method of converting biased structures into statically
optimal ones, and it enables us to waive the requirement of
distributional knowledge present in all previous biased point location
structures, at least for triangulations.
Our approach can be seen as a generalization and simplification of a method
by Goodrich for dictionaries~\cite{Goodrich00}.

\section{Notation}
Let $U$ be some universal set, and let $S$ be a
partition of $U$ into $n$ pieces. The elements
of $U$ are called \emph{points}, the subsets in $S$ are called \emph{regions}.
A \emph{location query} takes some point $p \in U$ and
returns the region $s \in S$ with $p \in s$. The result
of a location query with input $p$ is denoted by $q(p)$.
A data structure for location queries is called a \emph{location
query structure}. 

Let $P=\langle p_1, p_2, \dots , p_m \rangle$ 
be a sequence of $m$ queries, and 
let $Q \eqdef \langle q(p_1), q(p_2),$ $\dots, q(p_m) \rangle$ 
denote the results of these queries.
Let $f_t(s)$ be the number of occurrences of $s$ in the first $t$ elements
of $Q$, and define $f(s) \eqdef f_m(s)$, the number of
times $s$ occurs in the entire sequence. Furthermore, let $t_j(s)$ be
the time of the $j$\textsuperscript{th} occurrence of $s$ in $Q$; 
thus $f_{t_j(s)}(s)=j$.

We use $\log x$ to refer to $\max (1,\log_2 x)$; this avoids clutter generated 
by additive terms that would otherwise be needed to handle degenerate cases 
of our analysis.
We next define the notion of a biased  structure.
\begin{definition} 
Let $S$ be set of $n$ regions, 
and let $D$ be a location query structure for $S$.
We say that $D$ is \emph{biased} if the following holds: There
exists a function $c_D : \mathbb{N} \rightarrow \mathbb{N}$ such that
given any weight function $w : S \rightarrow \mathbb{R}^+$,
$D$ executes any query sequence $P$ in total time
\[ 
O\left( c_D(n)+ \sum_{s \in S} f(s) \log \frac{\sum_{r \in S} w(r)}{w(s)}
\right).
\]
The function $c_D$ is called the \emph{construction cost} of the
structure.  
\end{definition}

Suppose we choose
$w(s)$ proportional to the number of queries that return the given
region, e.g., $w(s):=f(s)+1$.
In this case, a biased location query structure
achieves an amortized query time that is 
(of the order of) the entropy $\mathcal{H}$ of the query 
distribution. As we argued in the introduction, this is optimal for our
decision-tree model. We now define the notion of static optimality.

\begin{definition} \label{def:statopt}
Let $S$ be set of $n$ regions, and
let $D$ be a location query structure for $S$.
We say that $D$ is \emph{statically optimal} if 
there exists a function $c_D : \mathbb{N} \rightarrow \mathbb{N}$
such that
$D$ executes any query sequence $P$ of length $m$ 
in total time\footnote{By convention, $f(s) \log (m/f(s)) \eqdef 0$ if
$f(s) = 0$.}
\[ 
O\left(c_D(n)+ \sum_{s \in S} f(s) \log \frac{m}{f(s)} \right).
\]
We call $c_D$ the \emph{construction cost} of $D$.
\end{definition}

Note that a statically optimal structure is given neither 
the frequency function $f$ nor any weights in advance, in
particular, the structure does not need to be static.

We provide a simple method for making a biased location query
structure statically optimal, assuming a few technical
conditions. The main such condition is that we should be able to construct
the biased query structure not just on the set $S$, but on any subset 
$S'$ of $S$. We require that a location query structure 
for $S'$ performs as quickly as a biased structure for $S$ when a 
region in $S'$ is queried, and that it reports failure in $O(\log n)$ time 
if the query lies outside of $S'$.
Formally:

\begin{definition}\label{def:subset_biased}
Let $S$ be set of $n$ regions, and let $D$ be a location
query structure. We call $D$ 
\emph{subset-biased} on $S$ if the following holds:
there exists a function $c'_D : \mathbb{N} \rightarrow \mathbb{N}$
such that given a subset $S' \subseteq S$ of size $n'$ and a weight
function $w': S' \rightarrow \mathbb{R}^+$, the structure $D$ executes
any query sequence $P$ of length $m$
in time
\[
\hspace{-6pt}O\hspace{-2pt}\left(\hspace{-2pt}c'_D(n')\hspace{-2pt}+ \hspace{-3pt}
\sum_{s' \in S'} f(s') \log \frac{\sum_{r' \in S'} w'(r')}{w'(s')}  +
\hspace{-2pt} \left(\hspace{-2pt}m - \sum_{s' \in S'}
f(s')\hspace{-2pt}\right) \hspace{-2pt} \log
n\hspace{-2pt}\right)\hspace{-3pt}.
\]
For each query $p \in P$, we require that $D$ reports the
region $s' \in S'$ with $p \in s'$, if it exists, 
and that $D$ reports
a failure otherwise.
The function
$c'_D$ is called the \emph{construction cost} of the structure.  
\end{definition}

\noindent 
Note that $m - \sum_{s' \in S'} f(s')$ is just the number of queries 
that result in failure. Given Definition~\ref{def:subset_biased}, we may 
now state our main theorem:

\begin{theorem} \label{main} 
Let $S$ be a set of $n$ regions.
Suppose we
have an $O(\log n)$ time location query structure 
on $S$ with construction cost $O(n)$ and a subset-biased structure on $S$
with construction cost $O(n' \log n')$. Then we can construct
a statically optimal structure on $S$ with construction cost $O(n)$.  
\end{theorem}

%----------------------------------------------------------------------------
\section{The transformation}

We now describe the construction for Theorem~\ref{main}. 
By assumption, we are given a set $S$ of $n$ regions, and we
have available an $O(\log n)$ time location query structure $D$ on $S$ with 
construction cost $O(n)$ as well as a 
subset-biased structure with construction 
cost $O(n' \log n')$.

\subsection{Description of the structure}

Let $\alpha$ and $\beta$ be two constants
such that $0<\beta<\alpha<1-\beta<1$ (e.g., $\alpha=1/2$ and
$\beta=1/3$).
The simple idea behind our transformation is as follows:
after every $n^\alpha$  queries, we build a subset-biased  
structure for the $n^\beta$ most commonly accessed regions,
in $O(n^\beta \log n) = o(n^\alpha)$ time.
We also keep a static $O(\log n)$ time structure as a backup for 
failed queries in the subset-biased structure.  
Formally, the structure has several parts:

\begin{enumerate}

\item A static $O(\log n)$ query time structure. 

\item \label{pointtwo} A structure that keeps track of how often
each region was queried and that is capable of reporting the 
$k$ most popular regions in
$O(k)$ time. Since in each step we increment the count
for a single region by $1$, we can easily maintain such a structure in 
linear space and constant time per update. 
(The additional space overhead can be made sublinear at the expense 
of determinism thorough the use of a streaming algorithm for the 
so-called \emph{heavy hitters} problem (e.g.,\cite{BerindeInCoSt10}). 
This shows that our transformation
is also useful in a context where additional space is at a
premium, for example for implicit data structures or when
the data resides in read-only 
memory~\cite{AsanoMuRoWa11}).

\item A subset-biased structure $D'$ that is built after 
$2n^\alpha$ queries and rebuilt every
$n^\alpha$\textsuperscript{th} query thereafter. The structure 
contains the at most 
$n^\beta$ most popular regions at the time of the rebuilding
that have been queried at least $2n^\alpha$ times.
In the choice of these regions, we break ties arbitrarily.
The weight of a region $s$, denoted $w'(s)$, is the number of queries to 
$s$ at the time of the
rebuilding. More precisely, if the rebuilding is at time $t$, we set 
$w'(s) \eqdef f_t(s)$.
Computing the $n^\beta$ most popular regions and the weight function
$w'$ takes time $O(n^\beta)$ with the structure from Part~\ref{pointtwo}.  
By assumption, the construction cost of $D'$
is $O(n^\beta \log n)$.  
\end{enumerate}

\noindent A search is executed on the subset-biased structure first. If it
fails (at amortized cost $O(\log n)$), it is executed in the 
static $O(\log n)$ time structure.

\subsection{Initial analysis of structure}

We will now analyze the properties of our structure.
Our first lemma describes a key property of the rebuilding process:
for any sufficiently popular region $s$, the amortized query time for $s$ is 
proportional to the amortized query time a biased structure
would achieve if it were weighted with the frequencies observed so far.

\begin{lemma} \label{singlequery} 
Consider the query $p_t$ at time $t$, and let $s \eqdef q(p_t)$ denote the 
resulting region.
Suppose that  $f_t(s) \geq 2 n^\alpha$. 
Then the amortized cost for query $p_t$ is
\[ 
O\left(  \log \frac{t}{f_t(s)-n^\alpha} \right).
\] 
\end{lemma}

\begin{proof} 
Since $f_t(s) \geq 2n^\alpha$, we have $t \geq 2n^\alpha$. Thus,
we first query the subset-biased structure $D'$.
Suppose that $D'$ has been rebuild last at time $t' \geq t - n^\alpha$. 
There are two cases. 

Suppose first that $s$ is
contained in $D'$.  Definition~\ref{def:subset_biased}
ensures that the amortized time for the query in $D'$ is 
$O(\log(W'/f_{t'}(s)))$, where $W'$ denotes the total number 
of queries for the regions in $D'$ at time
$t'$.
We have $W' \leq t$ (there have been $t$ queries so far) and 
$f_{t'}(s) \geq f_t(s) - n^\alpha$ (there have been at most $n^\alpha$
queries since rebuilding). The lemma follows.

Now suppose that $s$ is not in $D'$. In this case, the query 
takes $O(\log n)$ amortized time in $D'$ and $O(\log n)$ time in the
static structure. 
We know that at time $t'$, there were  $n^\beta$ regions 
at least as popular as $s$. Thus, $n^\beta f_{t'}(s)  \leq t' \leq t$. 
It follows that 
\[
\beta \log n = \log n^\beta \leq  \log \frac{t}{f_{t'}(s)} \leq 
\log \frac{t}{f_t(s) - n^\alpha},
\]
and the claimed bound 
suffices to account for the
$O(\log n)$ query time.
\end{proof}

\noindent
Using Lemma~\ref{singlequery}, we can now bound the running time
in terms of the query frequencies. 

\begin{lemma} \label{bigruntime} 
Let $S$ be a set of $n$ regions. Our structure executes any query sequence 
$P$ on $S$ of length $m$ in time
\[
\hspace{-3pt} O\hspace{-2pt} \left( \sum_{s \in S} \left(
\underbrace{\min(f(s), 2 n^\alpha) \log n}_{\text{first $2n^\alpha$ queries
to $s$}} + \hspace{-7pt} \underbrace{\sum_{j=2n^\alpha}^{f(s)} \log
\frac{t_j(s)}{f_{t_j(s)}(s)-n^\alpha}}_{\text{queries to $s$ after the
$2n^\alpha$th}} \right) + \underbrace{\left \lfloor
\frac{m}{n^\alpha} \right \rfloor n^\beta\log n}_{\substack{\text{rebuild
biased}\\\text{structure}}} + \hspace{-5mm}
\underbrace{n}_{\substack{
\text{static}\\ \text{structure}\\\text{construction}}} \right).
\]
\end{lemma}

\begin{proof} 
The main summation is over the regions in $S$. For each
region $s$, the initial $2 n^\alpha$ (or less) queries take 
time $O(\log n)$, since during these queries $s$ is never in 
the subset-biased structure. The running times for the 
remaining queries (if any) are 
bounded using Lemma~\ref{singlequery}. The first additional term 
comes from the $O(n^\beta\log n)$ construction cost of the 
subset-biased structure, incurred every $n^\alpha$ operations. 
The final term is the
linear one-time cost to build the static structure.
\end{proof}

\subsection{Technical Lemmas}

In order to simplify the bound in Lemma~\ref{bigruntime}, we need
two technical lemmas to deal with the various terms. 
The first lemma shows how to simplify the summation for the later queries.

\begin{lemma} \label{longmath} 
Let $S$ be a set of $n$ regions, and let $P$ be a query sequence on
$S$ of length $m$.
For each region $s \in S$, we have
\[ 
\sum_{j=2n^\alpha}^{f(s)} \log
\frac{t_j(s)}{f_{t_j(s)}(s)-n^\alpha} \leq  f(s) \left(3 + \log
\frac{m}{f(s)}\right).
\]
\end{lemma}

\begin{proof} 
Since $f_{t_j(s)}(s) = j \geq 2n^\alpha$, we
have $f_{t_j(s)}(s) - n^\alpha  \geq j/2$. Also,  $t_j(s) \leq m$. Thus,
\[ 
\sum_{j=2n^\alpha}^{f(s)} \log
\frac{t_j(s)}{f_{t_j(s)}(s)-n^\alpha} 
\leq \sum_{{j=1}}^{f(s)}
\log \frac{{m}}{j/2} = \log \frac{(2m)^{f(s)}}{f(s)!} \leq
\log \frac{(2em)^{f(s)}}{f(s)^{f(s)}} \leq f(s) \left(3+\log 
\frac{m}{f(s)} \right).
\]
Here, we used Stirling's formula to bound
$f(s)! \geq (f(s)/e)^{f(s)}$.
\end{proof}

The second lemma deals with the time for the initial queries.

\begin{lemma} \label{smallcase} 
Let  $\gamma$ be a constant with $\alpha<\gamma<1$.
If $m \geq n^{\gamma}$, then 
\[
\min(f(s), 2 n^\alpha) \log n = 
O \left( f(s)\log \frac{m}{f(s)} \right).
\]
\end{lemma}

\begin{proof} 
Set $\delta \eqdef (\alpha + \gamma)/2$.
If $f(s) \leq n^{\delta}$, the lemma holds since 
\[
f(s)\log (m/f(s)) \geq f(s) \log(m/n^{\delta}) 
= \Omega(f(s)\log n) = \Omega(\min(f(s), 2 n^\alpha) \log n).
\]
If $f(s) > n^{\delta}$, then 
\[
f(s)\log (m/f(s)) > n^{\delta}
\geq 2n^\alpha \log n,
\]
for $n$ large enough, as desired (recall that we defined $\log x$ to be at least $1$).
\end{proof}

\subsection{Main theorem} 

We can now prove our main theorem.

\begin{proof}[of Theorem~\ref{main}]
By Definition~\ref{def:statopt}, we need to prove that
the execution time is
\[ 
O \left( n +  \sum_{s \in S} f(s)\log \frac{m}{f(s)}\right).
\]
By Lemma~\ref{bigruntime}, the running time is bounded by
%\begin{align*} 
\[
\hspace{-9pt} O\hspace{-2pt}\left(\hspace{-2pt}
\rule{0cm}{21pt} \sum_{s \in S} \left( \min(f(s), 2 n^\alpha) \log n +
\hspace{-7pt} \sum_{j=2n^\alpha}^{f(s)} \hspace{-4pt}\log
\frac{t_j(s)}{f_{t_j(s)}(s)-n^\alpha}\right) + 
\left \lfloor \frac{m}{n^\alpha} \right \rfloor n^\beta\log n + n
\rule{0cm}{21pt} \right). 
%\end{align*}
\]
We now apply Lemma~\ref{longmath} and note that $\left \lfloor
\frac{m}{n^\alpha} \right \rfloor n^\beta\log n = o(m)$ to obtain 
a running time bound of
\[ 
O \left( \sum_{s \in S} \left( \min(f(s), 2 n^\alpha) \log n + f(s)
\left(3+ \log \frac{m}{f(s)}\right) \right) + n + m \right).
\]
Since we defined $\log x$ to be at least $1$, this simplifies to
\[ 
O \left( \sum_{s \in S} \left( \min(f(s), 2 n^\alpha) \log n + f(s)\log
\frac{m}{f(s)} \right) + n \right).
\]
If $m\leq n^{1-\beta}$, the sum over $s \in S$ is at most 
$n^{1-\beta}\log n = o(n)$. In this case, the bound simplifies 
to $O(n)$,  and the theorem is proved. Otherwise,
if $m > n^{1-\beta}$, Lemma~\ref{smallcase} applies with
$\gamma \eqdef 1 - \beta$ (a legal choice by our assumption on
$\alpha$ and $\beta$), and  the term
$\min(f(s), 2 n^\alpha) \log n$ collapses into $f(s)\log (m /f(s))$
to give the theorem.  
\end{proof}

\section{Point location}

\begin{theorem}\label{ppl} 
There is a data structure for point location
in a planar triangulation of size~$n$ that can execute any query sequence of
length $m$ in time 
\[ 
O \left( \sum_{s \in S}^n f(s)\log \frac{m}{f(s)} + n \right).
\] 
\end{theorem}

\begin{proof} 
It is easy to apply our general transformation to the problem of planar
point location in a triangulation, as all of the required ingredients
are well known.  We assume that the triangulation is given in a standard
representation, such as a doubly-connected edge list 
(e.g.,\cite[\S 2.2]{marks}). 

For the static structure with $O(\log n)$ query time and  $O(n)$ 
construction time, Kirkpatrick's algorithm~\cite{kirk} can
be used.  For the subset-biased structure, the provided subset of 
$n^\beta$ triangles may not be a connected triangulation and thus needs to
be triangulated; this takes time $O(n^\beta \log n)$ using the classic 
line sweep approach~\cite{LeePr77}. 
This creates $O(n^\beta)$ new triangles, which are marked specially and given
small weights. 
The resultant triangulation and weighting is given to a biased structure 
such as Iacono's~\cite{Iacono04}. The marking can be used to detect whether
a query to the subset-biased structure was successful. 
With all ingredients in hand, the claim now follows from Theorem~\ref{main}.
\end{proof}

Our choice of structures reflects a desire for the strongest asymptotic
bounds possible. Thus, we have avoided structures that are randomized or
that have non-linear construction cost; such structures, however, have far
superior constants than the ones we use. If we took a data structure for the
static $O(\log n)$ time queries with an $O(n \log n)$
construction cost instead of $O(n)$, this would simply change the linear 
additive term in Theorem~\ref{ppl} to $n \log n$.

\section{Point location in polygonal subdivisions with non-constant sized
cells}
Our work applies to point location in triangulations.
It can also be extended to polygonal subdivisions
where each region has constant complexity. Indeed, suppose
every region has $k+2$ edges. We can just triangulate each region and
then apply our result. As mentioned in the introduction,
this operation could increase the entropy of the query outcomes.
However, the \emph{log sum inequality}~\cite[Theorem~2.7.1]{CoverTh06}
implies that $\sum_{i=1}^{k} p_i \log(1/p_i) \leq p \log (k/p)$
for any nonnegative $p_1, p_2, \dots, p_k$ and $p = \sum_{i=1}^k p_k$.
Thus, if we subdivide a region with probability $p$ into
$k$ triangles, the entropy increases by at most $p \log k$. It
follows that the overall entropy grows by at most $\log k$, which 
is acceptable if $k$ is constant.

Recently,
several data structures have been developed for optimal point location where
the distribution is known in advance for convex connected
\cite{ColletteDuIaLaMo08}, connected
\cite{ColletteDuIaLaMo09}, and arbitrary polygonal
\cite{BoseDeDoDuKiMo10} subdivisions of the plane,
as well as the more general odds-on trees \cite{BoseDeDoDuKiMo10a}.
Unfortunately, these structures are not biased according to our
definition, since entropy-based lower bounds 
are not meaningful for them: a convex $k$-gon splits the
plane into two regions, so the entropy of the query outcomes is constant.
Nonetheless, some distributions require
$\Omega(\log n)$ time for a point location query (in a reasonable model
of computation that is more restrictive than the one described here).

The entropy-sensitive structures for
non-triangulations all basically work by triangulating the given subdivision as
a function of the provided probability distribution, and then using one of the
biased structures on the resultant triangulation. The main conceptual problem
in using our framework with such a structure is that it is  unclear how to
triangulate during the rebuilding process, since the optimal triangulation is
not known in advance. One could imagine that triangulating during each rebuild
based on the observed queries so far would work well, but proving this would
require a more complex and specialized analysis than what has been presented in
this paper.

\section*{Acknowledgments}

The second author would like to thank Pat Morin for suggesting the 
problem to him, for stimulating discussions on the subject,  and for 
hosting him during a wonderful stay at the Computational
Geometry Lab at Carleton University.
We would also like to thank the anonymous referees for 
insightful comments that helped improve the presentation of
the paper.

%\bibliographystyle{abbrv}
%\bibliography{bib}
\newcommand{\SortNoop}[1]{}\def\cprime{$'$}

\end{document}